\documentclass[runningheads,envcountsame]{llncs}
\pdfoutput=1
\usepackage{microtype}
\usepackage{graphicx}
\usepackage{amsmath}
\usepackage{amssymb}
\usepackage{color}
\usepackage{caption}
\usepackage{subfig}

\usepackage{algpseudocode}
\usepackage{algorithm}
\let\oldReturn\Return
\renewcommand{\Return}{\State\oldReturn}

\newcommand{\name}[1]{\textnormal{#1}}
\newcommand{\length}[0]{\ell}
\newcommand{\opt}[0]{\textnormal{OPT}}
\newcommand{\class}[1]{\textnormal{#1}}

\newcommand{\floor}[1]{\left\lfloor #1 \right\rfloor}
\newcommand{\ceil}[1]{\left\lceil #1 \right\rceil}

\newcommand{\statement}[3]{
\begin{list}{}{
\setlength{\leftmargin}{0.1in}
\setlength{\rightmargin}{0.1in}
\setlength{\parsep}{0pt}
\setlength{\itemsep}{2pt}
\setlength{\topsep}{\itemsep}
\setlength{\partopsep}{\itemsep}
}
\item
{\name{#1}}
\item
{INSTANCE: #2}
\item
{QUESTION: #3}
\end{list}
\vspace{1mm}
}

\graphicspath{{img/}}

\begin{document}

\title{Star Routing: Between Vehicle Routing and Vertex Cover}

\author{Diego {Delle Donne}\inst{1} \and Guido Tagliavini\inst{2}}

\authorrunning{D. Delle Donne \and G. Tagliavini}

\institute{Universidad Nacional de General Sarmiento, Instituto de Ciencias, Malvinas Argentinas, Argentina\\
\email{ddelledo@ungs.edu.ar}
           \and
	          Rutgers University, School of Arts and Sciences, New Jersey, US\\
	          \email{guido.tag@rutgers.edu}
}

\maketitle

\begin{abstract}
We consider an optimization problem posed by an actual newspaper company, which consists of computing a minimum length route for a delivery truck, such that the driver only stops at street crossings, each time delivering copies to all customers adjacent to the crossing. This can be modeled as an abstract problem that takes an unweighted simple graph $G = (V, E)$ and a subset of edges $X$ and asks for a shortest cycle, not necessarily simple, such that every edge of $X$ has an endpoint in the cycle.

We show that the decision version of the problem is strongly \class{NP-complete}, even if $G$ is a grid graph. Regarding approximate solutions, we show that the general case of the problem is \class{APX-hard}, and thus no PTAS is possible unless $\class{P} = \class{NP}$. Despite the hardness of approximation, we show that given any $\alpha$-approximation algorithm for metric \name{TSP}, we can build a $3\alpha$-approximation algorithm for our optimization problem, yielding a concrete $9/2$-approximation algorithm.

The grid case is of particular importance, because it models a city map or some part of it. A usual scenario is having some neighborhood full of customers, which translates as an instance of the abstract problem where almost every edge of $G$ is in $X$. We model this property as $|E - X| = o(|E|)$, and for these instances we give a $(3/2 + \varepsilon)$-approximation algorithm, for any $\varepsilon > 0$, provided that the grid is sufficiently big.


\end{abstract}

\section{Introduction}

Every morning, a well-known newspaper\footnote{Unfortunately, for confidentiality reasons, we cannot disclose their identity.} in Buenos Aires needs to deliver a copy to each subscriber by trucks. For now, assume there is only one truck. Traditionally, the truck stops in front of each customer's house, every time delivering a single copy of the paper. But now, the company thinks there could be a better (that is, cheaper) way to do it: instead of stopping to make a single delivery, the truck will only stop at street crossings, and each time the driver will pick up a pile of copies and deliver them to all customers located on any of the (typically four) adjacent streets. The goal is to minimize the number of blocks traveled by the truck.

We model the city topology as a simple graph, and the set of customers as a subset of edges. In other words, we distinguish blocks that have at least one customer, but we don't care if there is more than one customer in a single block. If $C$ is a cycle and $X$ is a subset of edges of a simple graph, we say that $C$ \textit{covers} $X$ if every edge of $X$ has an endpoint in $C$. The formal description of the problem is the following:

\statement{STAR ROUTING}
{A simple graph $G = (V, E)$, a non-empty subset of edges $X \subseteq E$, and a positive integer $K$.}
{Does $G$ have a cycle, not necessarily simple, of length at most $K$ that covers $X$?}

\noindent
Since all edges can be traversed in both directions, \name{STAR ROUTING} (or simply \name{STAR}) models all streets as two-way streets. Also, note that \name{STAR} doesn't ask about which (or how many) road crossings the truck should stop at and deliver during its journey.

\begin{figure}
\centering
\subfloat[A possible set of customers, marked as red dots, on a small part of Boedo neighborhood in Buenos Aires. The light blue area is an arbitrary boundary for the truck.]{
	\label{subfig:problem}
	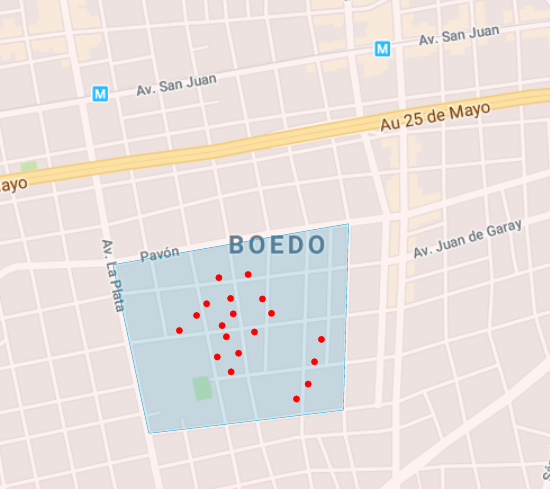
	}\hfill
\subfloat[A \name{STAR} instance that models the problem. Red edges are blocks that contain customers.]{
	\label{subfig:instance}
	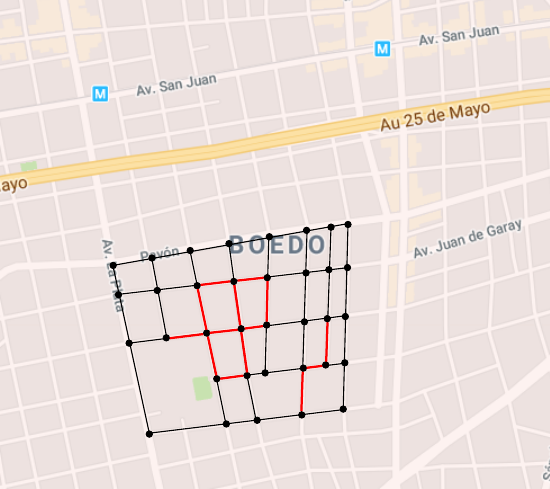
}
\caption{}
\label{fig:example}
\end{figure}

\begin{figure}
	\begin{center}
		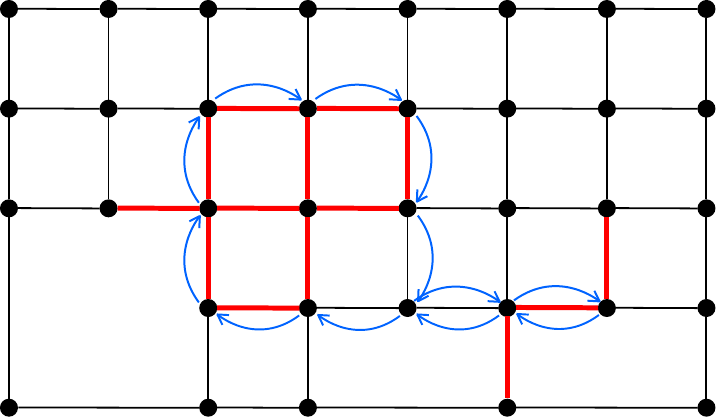
	\end{center}		
	\caption{Cleaned-up version of the \name{STAR} instance of Figure \ref{subfig:instance}. The arrows show a feasible solution.}
	\label{fig:solution}
\end{figure}

Consider the example of Figure \ref{subfig:problem}, which represents a real-life setting with subscribers shown as red dots. This is mapped to the \name{STAR} instance shown in Figure \ref{subfig:instance}. Each block that contains at least one customer is mapped as a red edge, and $X$ is the set of red edges. A feasible solution is presented in Figure \ref{fig:solution}. Indeed, this cycle of length $12$ given by the arrows is a feasible solution because every red edge has one endpoint in the cycle. In contrast, if we wanted to stop precisely at every customer's address, we would need to go through at least $16$ edges: one per red edge, plus $4$ more edges to move between the two connected components induced by the red edges. Thus, \name{STAR}'s solution is at least $40\%$ better, in terms of number of blocks traversed. This improvement may (or may not) be at the cost of greater overall time to perform the delivery, since now the driver has to walk from street crossings, carrying the newspapers. Clearly, the more packed the customers are, the better this alternative delivery model works, since a single vertex can cover many edges.

Keep in mind that a rigorous comparison between \name{STAR} and other delivery models is beyond the scope of this paper, as there are several practical considerations, like street orientation, speed limits or overall transit time, that we are not taking into account. Our focus is on studying \name{STAR}'s theoretical properties.

Despite newspaper delivery was the original motivation for this problem, it is worth noting that \name{STAR} may be applicable in other contexts as well, such as police patrol planning. In general, \name{STAR} captures characteristics from situations resembling covering problems but also involving vehicle routing features.

\paragraph*{\textbf{Related work.}}

A remarkable family of problems in combinatorial optimization are those known as vehicle routing problems (VRP). The basic component of a VRP are vehicles that move throughout a network, maybe starting and ending at some depot point, and moving between customers located over the network to deliver some sort of merchandise. The goal is usually minimizing some metric related to the total consumed time or the traveled distance. The origin of these problems can be traced back to the 1954 paper of Dantzig, Fulkerson and Johnson \cite{Da54}, in which they considered the \name{TSP}, which is a particular case of VRP. This work was followed by several other papers about the \name{TSP}. Clarke and Wright \cite{Cl64} added more than one vehicle to the problem, which led to the first proper formulation of VRP, though that name was not coined until the work of Golden, Magnanti and Nguyan \cite{Go77}.

In 1974, Orloff \cite{Or74} identified a class of routing problems of a single vehicle, which he called \name{GENERAL ROUTING PROBLEM} (\name{GRP}). The \name{GRP} takes a weighted graph $G = (V, E)$, and two sets $W \subseteq V$ and $F \subseteq E$, and asks to find a shortest cycle of $G$ that traverses every vertex in $W$ and every edge in $F$. This is a generalization of other well-known routing problems, like the \name{CHINESE POSTMAN PROBLEM} ($W = \emptyset$ and $F = E$), the \name{RURAL POSTMAN PROBLEM} ($W = \emptyset$), and the \name{TSP} ($G$ complete, $W = V$ and $F = \emptyset$). Notably, the first can be solved in polynomial time, whereas the decision versions of the latter two are \class{NP-complete} \cite{Le76}.

The \name{STAR} problem is a simple VRP with a single vehicle fleet, where we want to minimize the delivery cost, which we model as the total distance traveled by the vehicle. However, in contrast with traditional VRPs, the subset $X$ of edges containing customers can be covered just by visiting any of the two adjacent endpoints, rather than traveling along it. There are some variants of \name{TSP} with a similar flavor to that of \name{STAR}, in which the objective is to cover vertices with a more relaxed criteria than standard \name{TSP}. One of them is the \name{COVERING SALESMAN PROBLEM} (\name{CSP}) \cite{Cu89,Sh14} that takes a directed weighted graph and a positive integer $D$, and asks to find a minimum-length tour over a subset of vertices of $G$ such that every vertex not in the tour is within distance $D$ of some vertex in the tour. Current and Schilling \cite{Cu89} devised a simple heuristic for this problem, but its performance guarantee cannot be bounded due to the arbitrary weights. Interestingly, our approximation algorithm for the general version of \name{STAR} is similar to theirs, but since we assume unit weights we are able to derive a bound on the approximation ratio. Shaelaie et al. \cite{Sh14} presented metaheuristics for the \name{CSP} but, once again, they do not provide any theoretical guarantees. Another related problem is the \name{TSP WITH NEIGHBORHOODS} (\name{TSPN}) \cite{Ar94,Du03,De05}, that takes a set of regions in the euclidean plane, and asks for a shortest closed curve that visits each region. We should note that the grid version of \name{STAR}, which we will discuss later on, can be reduced to the rectilinear version of \name{TSPN}, in which each edge from $X$ is a region, but unfortunately the rectilinear \name{TSPN} has not been studied.

To the best of our knowledge, \name{STAR} hasn't been considered before, existing literature has little overlap with it, and it's the first VRP based on the notion of vertex cover.

\paragraph*{\textbf{Organization.}} In Section \ref{sec:complexity} we show that \name{STAR} is strongly \class{NP-complete}, even when the input graph is a grid. In Section \ref{sec:approximation_general} we study how well it is possible to approximate the general version of \name{STAR}. First, we give a lower bound by showing that \name{STAR} is \class{APX-hard}. Second, we provide a factory of approximation algorithms, which takes an $\alpha$-approximation algorithm for metric \name{TSP} and produces a $3\alpha$-approximation algorithm for \name{STAR}. This yields a $9/2$ approximation factor for the general case, and a $3 + \varepsilon$ factor when the graph is planar. In Section \ref{sec:approximation_grid} we develop a $(3 / 2 + \varepsilon)$-approximation algorithm for grid graphs, assuming there are asymptotically more edges with customers than not and that the grid is large enough. Finally, in Section \ref{sec:open} we state some open problems.

\paragraph*{\textbf{Notation.}} Let $\Pi$ be an optimization problem. Let $I$ be a valid input of $\Pi$. We write $\Pi^*(I)$ the value of an optimal solution of the problem $\Pi$ for $I$.

If $A$ is a finite set, $|A|$ is the cardinality of $A$. We denote $K(A)$ the complete graph whose set of vertices is $A$.

All graphs we consider in this paper are simple. All cycles and paths we consider are not necessarily simple. If $G$ is a graph, $\tau(G)$ is the cardinality of any minimum vertex cover of $G$. If $X$ is a subset of edges of $G$, $G[X]$ is the subgraph of $G$ induced by $X$. If $S$ is a path of $G$, $\length(S)$ is the number of edges of $S$ (counting repetitions). If the edges of $G$ have weights given by a function $w$, $\length_w(S)$ is the sum of the weights of the edges of $S$ (counting repetitions). If $T$ is another path of $G$, that starts where $S$ ends, $S \circ T$ is the path we get from first traversing $S$ and then $T$. If $u$ and $v$ are two vertices of $G$, $d_G(u, v)$ is the minimum $\ell(S)$ over every path $S$ between $u$ and $v$. If $p$ and $q$ are two points in $\mathbb{R}^2$, $d_1(p, q)$ is the Manhattan distance between them.

A grid graph with $n$ rows and $m$ columns is the cartesian product of graphs $P_n \square P_m$, where $P_k$ is the path of $k$ vertices. A star graph is a complete bipartite graph $K_{1, n}$, for some $n \geq 1$.

\section{\name{STAR} is hard, even for grids}
\label{sec:complexity}

In this section we show that \name{STAR} is \class{NP-complete} when we restrict $G$ to the class of grid graphs. We call this version of the problem \textit{grid \name{STAR}}. These instances are of practical interest, since grids are the most simple way of modelling a city layout. In particular, the problem is hard for planar graphs and for bipartite graphs, among all superclasses of grid graphs.

To prove completeness, we will reduce from the rectilinear variant of \name{TSP}. Recall the \name{TSP} takes a set of elements $S$ equipped with weights between each pair of elements, and a positive integer $L$, and asks if there exists a hamiltonian cycle in $K(S)$ with total weight $L$ or less. In the rectilinear version, the input is a set of points $P$ in the plane, with positive integer coordinates, and a positive integer $L$, and asks if $K(P)$ has a hamiltonian cycle with total Manhattan distance length $L$ or less.

The rectilinear \name{TSP} is \class{NP-complete}. In 1976, Garey et al. \cite{Ga76} proved this, by reducing from \name{EXACT COVER BY 3-SETS} (\name{X3C}), which takes a family $\mathcal{F} = \{F_1, \dots, F_t\}$ of $3$-element subsets of a set $U$ of $3n$ elements, and asks if there exists a subfamily $\mathcal{F'} \subseteq \mathcal{F}$ of parwise disjoint subsets such that $\cup_{F \in \mathcal{F'}} = U$. Since \name{X3C} has no numerical arguments, it is strongly \class{NP-complete}. The rectilinear \name{TSP} instance they build is such that both coordinates of every point in the set $P$, as well as the optimization bound $L$, are bounded by a polynomial on the size of the \name{X3C} instance. Thus, rectilinear \name{TSP} is strongly \class{NP-complete}.

The transformation we will use has a similar flavor than the one devised by Demaine and Rudoy \cite{De17} to show that solving a certain puzzle is \class{NP-complete}.

\begin{theorem}
Grid \name{STAR} is strongly \class{NP-complete}.
\end{theorem}
\begin{proof}
Given a cycle of $G$ it's easy to check in polynomial time if it covers all edges in $X$, and if it has length $K$ or less. Thus, the problem is in \class{NP}.

Now we reduce from rectilinear \name{TSP}. Let $P = \{p_1, \dots, p_n\}$ and the bound $L$ be an instance of rectilinear \name{TSP}. Let $m$ be the maximum coordinate of any point in $P$, so that all points lie in the rectangle $[1, m] \times [1, m]$. Let $c = 2(n + 1)$. We will build a grid graph $G$ by taking the $m \times m$ rectangular grid of points with lower left corner at $(1, 1)$, and expanding it by a factor of $c$. Formally, if $G = (V, E)$, then $V$ is the set of all integer coordinates points in $[c, cm] \times [c, cm]$, and $E$ is the natural set of edges we need to produce a grid out of $V$. Note that $cp_i \in V$ for every $i$. That is, multiplying by $c$ we map points from $P$ to $V$.

Let $e_i$ be any adjacent edge to $cp_i$ in $G$, and let $X = \{e_1, \dots, e_n\}$. Finally, let $K = cL$.

\begin{figure}
	\begin{center}
		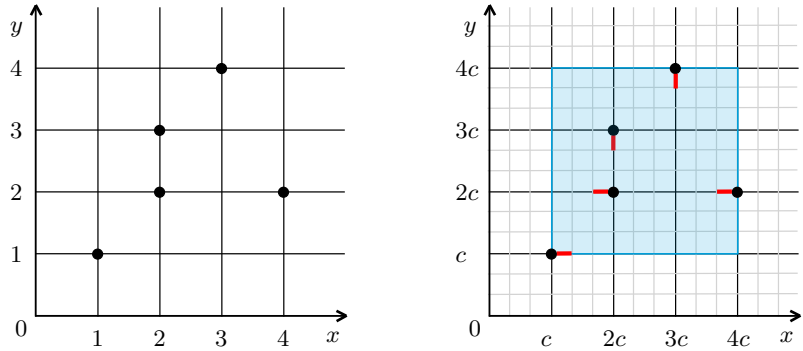
	\end{center}		
	\caption{Mapping an instance of rectilinear \name{TSP} (on the left) to grid \name{STAR} (on the right). The marked points on the first grid are $p_i$s, which are mapped to the second grid as $cp_i$. The light blue area denotes the graph $G$. The red edges make up the set $X$.}
	\label{fig:reduction}
\end{figure}

\paragraph*{\textbf{Polynomial time.}}

Since rectilinear \name{TSP} is strongly \class{NP-complete}, we can assume $m$ and $L$ are polynomial. The grid $G$ has size $O((cm)^2)$, which is polynomial because both $c$ and $m$ are. The coordinates of every vertex are bounded by $O(cm)$. Computing $X$ is obviously polynomial. Finally, computing $K = cL$ is also polynomial, since $L$ is polynomial. Thus, the reduction takes polynomial time, and every numerical value is bounded by a polynomial in the transformation's input size.

\paragraph*{\textbf{Rectilinear \name{TSP} to grid \name{STAR}.}}

Assume there is a hamiltonian cycle with Manhattan distance length $L$ or less, in $K(P)$. W.l.o.g., suppose $T = \langle p_1, \dots, p_n, p_{n + 1} = p_1 \rangle$ is such a cycle. For each $i$, let $S_i$ be a shortest path in $G$ from $cp_i$ to $cp_{i + 1}$. Then $S = S_1 \circ \dots \circ S_n$ is a cycle in $G$ that goes through every vertex $cp_i$, and thus covers $X$. Its length is
\[\ell(S) = \sum_{i = 1}^n \ell(S_i)
= \sum_{i = 1}^n d_1(cp_i, cp_{i + 1})
= c\sum_{i = 1}^n d_1(p_i, p_{i + 1})
= c\ell_{d_1}(T)
\leq cL = K
\]

\paragraph*{\textbf{Grid \name{STAR} to rectilinear \name{TSP}.}}
Suppose there is a cycle $S = \langle s_1, \dots, s_m, s_{m + 1} = s_1 \rangle$ of length $K$ or less that covers $X$ in $G$. At some point while we traverse $S$, we must get close to each $cp_i$, since the cycle covers $e_i$. More specifically, there exists an index $1 \leq j_i \leq m$ such that either $s_{j_i}$ is exactly $cp_i$, or $e_i = (cp_i, s_{j_i})$. This implies that $d_1(cp_i, s_{j_i}) \leq 1$. Assume w.l.o.g. that $j_1 \leq \dots \leq j_n$, since otherwise we can rearrange the indexes of the points $cp_i$. Consider the hamiltonian cycle $T = \langle p_1, \dots, p_n, p_{n + 1} = p_1 \rangle$ of $K(P)$. (Define, for convenience, $j_{n + 1} := j_1$.) We need to show that $\ell_{d_1}(T) \leq L$. Since $\ell_{d_1}(T)$ is an integer, it suffices to prove $\ell_{d_1}(T) < L + 1$. We start by rewriting
\[
\ell_{d_1}(T) = \sum_{i = 1}^n d_1(p_i, p_{i + 1}) = (1/c) \sum_{i = 1}^n d_1(cp_i, cp_{i + 1})
\]

\noindent
Since $d_1$ is a metric, we can decompose
\[d_1(cp_i, cp_{i + 1}) \leq d_1(cp_i, s_{j_i}) + d_1(s_{j_i}, s_{j_{i + 1}}) + d_1(s_{j_{i + 1}}, cp_{i + 1}) \leq 2 + d_1(s_{j_i}, s_{j_{i + 1}})\]

\noindent
Therefore
\[\ell_{d_1}(T) \leq (1/c) \left(2n + \sum_{i = 1}^n d_1(s_{j_i}, s_{j_{i + 1}}) \right)\]

\noindent
Consider the subpaths $S_i := \langle s_{j_i}, s_{j_i + 1}, \dots, s_{j_{i + 1}} \rangle$ (here we are using the fact that the indexes $j_i$ are ordered). Then $d_1(s_{j_i}, s_{j_{i + 1}}) \leq \ell_{d_1}(S_i)$. Since these subpaths are disjoint pieces of $S$, we have $\sum_{i = 1}^n \ell_{d_1}(S_i) \leq \ell_{d_1}(S) = \ell_{d_G}(S)$, so
\[\ell_{d_1}(T) \leq (1/c) \left(2n + \ell_{d_G}(S) \right) \leq 2n / c + K / c = n / (n + 1) + L < 1 + L \]

\noindent
as desired.
\qed
\end{proof}

\section{An approximation algorithm for the general case}
\label{sec:approximation_general}

Since \name{STAR} is a hard problem in regards to finding exact solutions, we investigate approximation algorithms. We start by showing that the general version of the problem is hard to approximate within a constant factor arbitrarily close to $1$. For this, we reduce from approximating the \name{VERTEX COVER} (\name{VC}) problem, which is known to be APX-hard \cite{Di04}. Given a simple graph $G$, \name{VC} asks for a minimum cardinality vertex cover of $G$.

\begin{theorem}
For every $\alpha$-approximation algorithm for \name{STAR} there is an $\alpha$-approximation algorithm for \name{VC}.
\end{theorem}
\begin{proof}
Let $A_{\name{STAR}}$ be an $\alpha$-approximation algorithm for \name{STAR}. Given an input graph $G = (V, E)$, the approximation algorithm for \name{VC} proceeds as follows. If $E = \emptyset$, return an empty set. If $G$ is a star graph, return the central vertex. Otherwise, every feasible vertex cover has two or more vertices. Consider the instance $I = (K(V), E)$ of \name{STAR}, that is, a complete graph where the set of customers are the edges of $G$. The algorithm computes $S = A_{\name{STAR}}(I)$ and outputs $S$ as a set.

The algorithm is polynomial, since we can construct $I$ in polynomial time. Note that every cycle in $K(V)$ that covers $E$ induces a vertex cover of $G$, and therefore $S$ is a feasible vertex cover of $G$. Reciprocally, every vertex cover of $G$ induces a cycle in $K(V)$ that covers $E$ (by fixing any order among the vertices in the cover), which implies that $\name{STAR}^*(I) \leq \tau(G)$. Since $S$ is an $\alpha$-approximation, we have $|S| \leq \ell(S) \leq \alpha \text{ } \name{STAR}^*(I) \leq \alpha \text{ } \tau(G)$.
\qed
\end{proof}

Dinur and Safra showed that it's hard to approximate \name{VC} within a factor $1.3606$ of optimal \cite{Di04}. Thus, \name{STAR} is hard to approximate as well.

\begin{corollary}
It's \class{NP-hard} to approximate \name{STAR} within a factor $1.3606$ of optimal.
\end{corollary}

Therefore, \name{STAR} doesn't admit a PTAS unless $\class{P} = \class{NP}$, and thus the best we can hope for is some constant-factor approximation algorithm. Indeed, we now show that \name{STAR} admits one.

During the rest of this paper, we denote $(G, X)$ an instance of \name{STAR}, and write $\opt := \name{STAR}^*(G, X)$. Recall that $X \neq \emptyset$.

\begin{lemma}
\label{lem:star_tau_inequality}
If $G[X]$ is not a star graph, then $\opt \geq \tau(G[X])$.
\end{lemma}
\begin{proof}
	Let $S$ be an optimal solution of $\name{STAR}$. Since $S$ covers $X$, we can extract a vertex cover $C$ of $G[X]$ from the set of vertices of $S$. Since $G[X]$ is not a star, it's easy to see that $S$ has two or more vertices, and thus $\length(S) \geq |C|$. Hence, $\opt = \length(S) \geq |C| \geq \tau(G[X])$.
\qed
\end{proof}

From now on, we assume $G[X]$ is not a star. It's easy to both recognize a star graph and, in that case, return the optimal solution (the central vertex of the star) in polynomial time.

\begin{lemma}
\label{lem:feasible_tsp_to_feasible_star}
Let $C$ be a vertex cover of $G[X]$. Starting from a feasible solution $T$ of $\name{TSP}$ for $(C, d_G)$ we can build, in polynomial time in $G$, a feasible solution $S$ of $\name{STAR}$ for $(G, X)$, such that $\length(S) = \length_{d_G}(T)$.
\end{lemma}
\begin{proof}
Let $T = \langle t_1, \dots, t_n, t_{n + 1} = t_1 \rangle$. Let $S_i$ be any shortest path between $t_i$ and $t_{i + 1}$, in $G$. Consider the path $S = S_1 \circ \dots \circ S_n$ of $G$, which covers $X$, since it traverses every vertex in $C$. This path can be computed in polynomial time, since it's the union of a polynomial number of shortest paths of $G$. We have $\length(S) = \sum_{i = 1}^{n}\length(S_i) = \sum_{i = 1}^{n}d_G(t_i, t_{i + 1}) = \length_{d_G}(T)$.
\qed
\end{proof}

Recall the classic $2$-approximation for \name{VC}, shown in Algorithm \ref{alg:vertex_cover_approximation}. We will refer to it as the \textit{approximation via matching}.

\begin{algorithm}
  \caption{\name{VC} $2$-approximation via matching}
  \label{alg:vertex_cover_approximation}
  \begin{algorithmic}[1]
  	\Require A simple graph $G$.
  	\State Compute any maximal matching $M$ of $G$. Let $M = \{(u_1, v_1), \dots, (u_m, v_m)\}$.
	\Return $\{u_1, \dots, u_m, v_1, \dots, v_m\}$
  \end{algorithmic}
\end{algorithm}

\begin{theorem}
\label{thm:bound_tsp_star}
Let $C$ be a vertex cover of $G[X]$, built with the approximation via matching. Then $\name{TSP}^*(C, d_G) \leq 3 \text{ } \opt$.
\end{theorem}
\begin{proof}

Let $C = \{u_1, \dots, u_m, v_1, \dots, v_m\}$, such that each $e_i := (u_i, v_i)$ is an edge of the maximal matching. Let $S = \langle s_1, \dots, s_n, s_{n + 1} = s_1 \rangle$ be an optimal solution of $\name{STAR}$ for $(G, X)$. The key observation is that since $e_i \in X$, and $S$ covers $X$, at least one of $u_i$ or $v_i$ is in $S$. W.l.o.g., assume $u_i$ is in $S$. Hence, for each $u_i$, there exists an index $1 \leq j_i \leq n$ such that $u_i = s_{j_i}$ (we define $j_{m + 1} := j_1$). W.l.o.g., assume that $j_1 \leq \dots \leq j_m$, since otherwise we can rearrange the elements of $C$ to satisfy it. Given this ordering, consider $T = \langle u_1, v_1, u_2, v_2, \dots, u_m, v_m, u_{m + 1} = u_1 \rangle$, which is a feasible solution of $\name{TSP}$ for $(C, d_G)$. It suffices to show that $\length_{d_G}(T) \leq 3 \text{ } \opt$. Figure \ref{fig:approximation} shows the sets and cycles defined so far.

\begin{figure}
	\begin{center}
		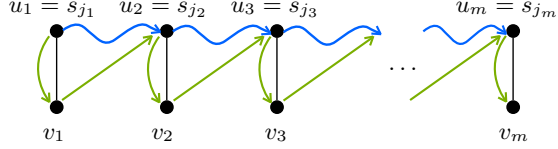
	\end{center}		
	\caption{Relation between $C$, $S$, and $T$. The curly blue arrows denote $S$, and the green arrows denote $T$. We do not show the edges that close the cycle. Also, $S$ may contain $v_i$s, but we don't illustrate this.}
	\label{fig:approximation}
\end{figure}

We have that
\[\ell_{d_G}(T) = \sum_{i = 1}^m(d_G(u_i, v_i) + d_G(v_i, u_{i + 1})) = \sum_{i = 1}^m(1 + d_G(v_i, u_{i + 1}))\]

\noindent
Since $d_G$ is a metric, $d_G(v_i, u_{i + 1}) \leq d_G(v_i, u_i) + d_G(u_i, u_{i + 1}) = 1 + d_G(u_i, u_{i + 1})$. Hence,
\[\ell_{d_G}(T) \leq \sum_{i = 1}^m(2 + d_G(u_i, u_{i + 1})) = 2m + \sum_{i = 1}^m d_G(u_i, u_{i + 1})\]

\noindent
Recall that $u_i = s_{j_i}$ for each $i$. Consider the subpaths $S_i := \langle s_{j_i}, s_{j_i + 1}, \dots, s_{j_{i + 1}}\rangle$. Then, $d_G(u_i, u_{i + 1}) = d_G(s_{j_i}, s_{j_{i + 1}}) \leq \length(S_i)$, and therefore
\[
\ell_{d_G}(T) \leq 2m + \sum_{i = 1}^m\length(S_i) \leq 2m + \length(S) = 2m + \opt
\]

\noindent
Since $C$ is a $2$-approximation, $2m = |C| \leq 2 \text{ } \tau(G[X])$. Finally, we use Lemma \ref{lem:star_tau_inequality} to get $|C| \leq 2 \text{ } \opt$, and we arrive to the desired bound.
\qed
\end{proof}

The proposed approximation algorithm for \name{STAR} is shown in Algorithm \ref{alg:star_approximation}. Note that the instance $(C, d_G)$ of \name{TSP} that $A_{\name{TSP}}$ approximates is, indeed, a metric instance, because $d_G$ is a metric.

\begin{algorithm}
  \caption{Approximation algorithm for \name{STAR}}
  \label{alg:star_approximation}
  \begin{algorithmic}[1]
  	\Require An instance $(G, X)$ of \name{STAR}.
  	\State Let $A_{\name{VC}}$ be the approximation via matching algorithm. Let $A_{\name{TSP}}$ be an approximation algorithm for metric \name{TSP}.
  	\State Compute $C = A_{\name{VC}}(G[X])$.
	\State Compute $T = A_{\name{TSP}}(C, d_G)$.
	\State Using $T$, build $S$ as in Lemma \ref{lem:feasible_tsp_to_feasible_star}.
	\Return $S$
  \end{algorithmic}
\end{algorithm}

\begin{theorem}
\label{th:star_approximation_algorithm}
If $A_{\name{TSP}}$ is an $\alpha$-approximation algorithm for metric \name{TSP}, then Algorithm \ref{alg:star_approximation} is a $3\alpha$-approximation algorithm for \name{STAR}.
\end{theorem}
\begin{proof}
The algorithm is polynomial, because each step is polynomial. The answer $S$ is a feasible solution of \name{STAR}, as stated in Lemma \ref{lem:feasible_tsp_to_feasible_star}. Regarding the performance guarantee,
\begin{align*}
\length(S) &= \length_{d_G}(T) &\text{(Lemma \ref{lem:feasible_tsp_to_feasible_star})}\\
&\leq \alpha \text{ } \name{TSP}^*(C, d_G) &\text{(} A_{\name{TSP}} \text{ is an } \alpha \text{-approximation)}\\
&\leq 3\alpha \text{ } \opt & \text{(Theorem \ref{thm:bound_tsp_star})}
\end{align*}
\qed
\end{proof}

Using Christofides' $3/2$-approximation algorithm for metric \name{TSP} \cite{Ch76}, we get the following concrete algorithm.

\begin{corollary}
\label{cor:planar_star}
There is a $9/2$-approximation algorithm for \name{STAR}.
\end{corollary}

If $G$ is restricted to some subclass of graphs, we could use a more specific approximation algorithm $A_{\name{TSP}}$ (one that doesn't work for all metric instances), and get a better approximation guarantee. For example, if $G$ is a planar graph (for instance, if $G$ is a grid graph), then we can use a PTAS \cite{Gr95}.

\begin{corollary}
For every constant $\varepsilon > 0$, there is a $(3 + \varepsilon)$-approximation algorithm for planar instances of \name{STAR}.
\end{corollary}

\section{An approximation algorithm for grids full of customers}
\label{sec:approximation_grid}

A typical and desired case in the newspaper delivery business is having neighborhoods full of customers. We model such a dense neighborhood with a grid graph, where almost every edge is in $X$. In this section, we propose a method to approximate the optimal solution, tailored for this dense setting.

The key idea is that since almost every edge is in $X$, any feasible solution will cover almost every edge of $E$. What if instead of covering just $X$, we cover the whole set $E$? We show that if $|E - X| = o(|E|)$, then there is such a na\"ive tour that is guaranteed to have length at most a factor $3 / 2 + \varepsilon$ of the optimal, for sufficiently large grids.

A cycle that covers every edge in a graph is somewhat similar to the concept of \textit{space-filling curve}. Mathematically, a space-filling curve is a curve whose range contains a certain $2$-dimensional area, for example the unit square. Space-filling curves have been used before to compute tours for the \name{TSP}. In 1989, Platzman and Bartholdi \cite{Pl89} proved that if we visit the vertices in the order given by a specific space-filling curve, we get an $O(\log n)$-approximation algorithm. In the graph-theoretical setting of \name{STAR}, \textit{filling} means to cover edges, but not necessarily to visit every vertex. Our dense-case approximation can be thought of as a \textit{space-filling cycle}.

Before constructing this particular cycle we prove some auxiliary results that will help us to analyze its performance.

\begin{lemma}
Let $e$ be an edge of a graph $G$. Then $\tau(G) \leq \tau(G - e) + 1$.
\end{lemma}
\begin{proof}
If we take any vertex cover of $G - e$ and add one of the endpoints of $e$ (if not already in the vertex cover), we get a vertex cover of $G$.
\qed
\end{proof}

In what follows, we will write $\overline{X} := E - X$.

\begin{lemma}
\label{lem:complement_bound}
Let $(G, X)$ be an instance of \name{STAR}, such that $G[X]$ is not a star graph. Then $\tau(G) \leq \opt + |\overline{X}|$.
\end{lemma}
\begin{proof}
If we repeatedly apply the previous lemma, each time subtracting a new vertex of $\overline{X}$, we get $\tau(G) \leq \tau(G - \overline{X}) + |\overline{X}| = \tau(G[X]) + |\overline{X}|$. Using Lemma \ref{lem:star_tau_inequality} we arrive to the desired inequality.
\qed
\end{proof}

The proof plan is to construct a space-filling cycle, compare its length with $\tau(G)$ and then use Lemma \ref{lem:complement_bound} to bound its performance. It will come in handy to know the exact value of $\tau(G)$ when $G$ is a grid graph.

\begin{lemma}
Let $G$ be a grid graph with $n$ rows and $m$ columns. Then $\tau(G) = \floor{nm / 2}$.
\end{lemma}
\begin{proof}
($\leq$) Note that $G$ is bipartite. Consider any bipartition of its vertices. Both subsets of the partition are vertex covers, and since there are $nm$ vertices in total, one of them must have size at most $\floor{nm / 2}$.

($\geq$) We use the fact that the size of any matching is always less than or equal to the size of any vertex cover. It suffices to exhibit a matching of size $\floor{nm / 2}$. To build such a matching, we go over every other row, and for each one we take every other horizontal edge. If $m$ is odd, we also take every other vertical edge of the last column. It's clear that this is a matching, and it's a matter of simple algebra to verify that it has $\floor{nm / 2}$ edges.
\qed
\end{proof}

We are ready to exhibit and analyze our space-filling cycle.

\begin{theorem}
\label{thm:approximation_grid_constant}
There is an approximation algorithm for grid \name{STAR} that computes solutions with length at most $(3/2 + O(1/n + 1/m))(\opt + |\overline{X}|)$, where $n$ and $m$ are the number of rows and columns, respectively, of the input grid graph.
\end{theorem}
\begin{proof}
We introduce some terminology to describe the cycle. Enumerate the grid's rows from $1$ to $n$, being $1$ the uppermost row and $n$ the lowest one. We divide the grid into horizontal \textit{stripes}, such that the $i$-th stripe, $i \geq 1$, consists of rows $2i - 1$ and $2i$. If $n$ is odd, the last stripe is formed only by the last row.

First we sketch a high-level description. Starting from the upper left corner, we will visit the stripes in order. Initially we move right, until we get to the right border of the grid, the end of the first stripe. Then we go down to the second stripe, and now move left until we get to the left border. Next we go down to the third stripe. The process continues until we finish visiting the last stripe. If the last one is a single row, we move in a straight line. Finally, we go back to the starting position.

More specifically, on stripe $i$, for some odd $i$, we move from left to right following a square wave pattern, which we call \textit{period}. A period is a sequence of the following single-edge moves: down, right, right, up, right, right. This is illustrated in Figure \ref{subfig:period}. We repeat this sequence of moves until it's no longer possible, at the right border of the grid. At this point, we could be anywhere between the beginning and the end of a period. In any case, we stop, and move exactly two edges down. On stripe $i + 1$ we move in the opposite direction, from right to left, repeating the steps we did on stripe $i$, but in reverse order. When we get to the left border, we go down two edges again, and we are ready to repeat the process. When we reach the end of the grid, we close the cycle by adding a shortest path to the initial vertex. An example of this construction is shown in Figure \ref{subfig:example}.

\begin{figure}
\centering
\subfloat[A period.]{
	\label{subfig:period}
	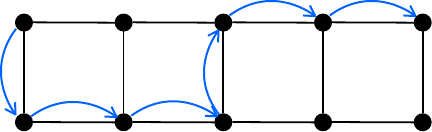
	}\hfill
\subfloat[The tour for a $5 \times 7$ grid. The blue arrows show how each stripe is traversed. The gray arrows show how the cycle goes from one stripe to the next one. The path between the lower right corner and the upper left corner, that closes the cycle, is not drawn for clarity.]{
	\label{subfig:example}
	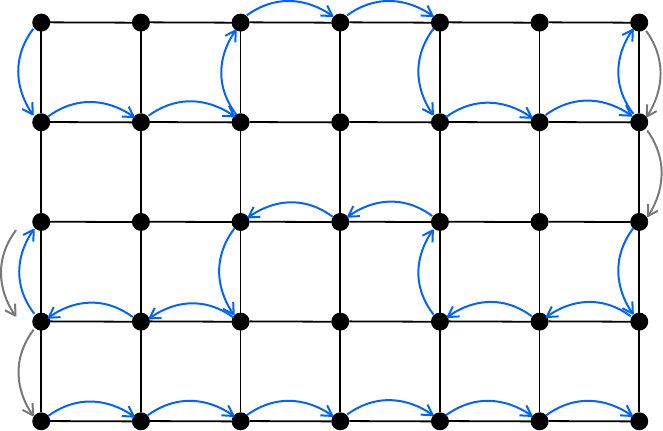
}
\caption{}
\label{fig:test}
\end{figure}

Let $C$ be this cycle. It's easy to see that $C$ covers each edge of $G$, and that it can be computed in polynomial time. Our approximation algorithm simply outputs $C$. We now show that $\ell(C) \leq (3/2 + O(1/n + 1/m)) \tau(G)$. By Lemma \ref{lem:complement_bound}, this implies the desired bound.

Each of the $\floor{n / 2}$ two-rows stripes contains $m - 1$ horizontal and $\ceil{m / 2}$ vertical edges of $C$. To move between two consecutive two-rows stripes, $C$ uses exactly $2$ edges. Additionally, if $n$ is odd, the last stripe is a single row, and we account $m - 1$ edges for moving along that row, plus $2$ edges to move from the previous stripe. Finally, we have at most $n + m - 2$ extra moves to go from the last stripe to the initial position. Summing everything,
\[\ell(C) \leq \floor{n / 2}(m - 1 + \ceil{m / 2}) + (\floor{n / 2} - 1)\text{ }2 + (m - 1 + 2) + (n + m - 2)\]

\noindent
The first term accounts for intra-stripes moves, the second for inter-stripes moves, the third for a potential single-row stripe, and the last one for the cost to go back to the initial position. A sloppy bounding of the floor and ceiling functions yields
\begin{align*}
\ell(C) &\leq (n / 2)(m - 1 + m / 2 + 1) + 2(n / 2 - 1) + (m + 1) + (n + m - 2)\\
&= (3/2)(nm / 2) + 2n + 2m - 3\\
&\leq (3/2)\tau(G) + 2n + 2m - 9 / 4\\
&= (3 / 2 + O(1 / m + 1 / n)) \tau(G)
\end{align*}
\qed
\end{proof}

\begin{corollary}
\label{cor:star_approximation_grid}
There is an approximation algorithm for grid \name{STAR} such that for every $\varepsilon > 0$, there exist positive numbers $n_{\varepsilon}$ and $m_{\varepsilon}$ such that the algorithm computes solutions with length at most $(3/2 + \varepsilon)(\opt + |\overline{X}|)$, for every input grid with $n \geq n_{\varepsilon}$ rows and $m \geq m_{\varepsilon}$ columns.
\end{corollary}

Recall that we are interested in the case where almost all edges belong to $X$, that is, $|E - X| = |\overline{X}| = o(|E|)$. As we can see, the smaller the $|\overline{X}|$, the better the approximation, showing that the space-filling cycle is a promising strategy for the dense readership case.

\begin{theorem}
There is an approximation algorithm for grid \name{STAR} such that for every $\varepsilon > 0$, there exist positive numbers $n_{\varepsilon}$ and $m_{\varepsilon}$ such that the algorithm is $(3/2 + \varepsilon)$-approximated, for every input grid with $n \geq n_{\varepsilon}$ rows, $m \geq m_{\varepsilon}$ columns and $|\overline{X}| = o(|E|)$.
\end{theorem}
\begin{proof}
If $G$ is a grid graph, then $\tau(G[X]) \geq |X| / 4$, because a single vertex can cover up to $4$ edges. Hence, $\opt \geq \tau(G[X]) = \Omega(|X|)$.

Since $|\overline{X}| = o(|E|)$, we have $|X| = \Theta(|E|)$, and thus $\opt = \Omega(|X|) = \Omega(|E|)$. This in turn implies that $|\overline{X}| = o(\opt)$, which means that for all $\varepsilon > 0$ there exist positive integers $n_{\varepsilon}$, $m_{\varepsilon}$ such that $|\overline{X}| < \varepsilon \text{ } \opt$ for every $n \geq n_{\varepsilon}$ and $m \geq m_{\varepsilon}$.

Fix any $\varepsilon > 0$. Let $\varepsilon_1, \varepsilon_2 > 0$ be any two positive reals such that $\varepsilon_1 + (3/2) \varepsilon_2 + \varepsilon_1 \varepsilon_2 \leq \varepsilon$. Instantiate Corollary \ref{cor:star_approximation_grid} with $\varepsilon_1$, and let $n_{\varepsilon_1}$ and $m_{\varepsilon_1}$ be the minimum numbers of rows and columns, respectively. Let $n_{\varepsilon_2}, m_{\varepsilon_2}$ be such that if $n \geq n_{\varepsilon_2}$ and $m \geq m_{\varepsilon_2}$, then $|\overline{X}| < \varepsilon \text{ } \opt$.

Under these choices, if $n \geq n_{\varepsilon} = \max\{n_{\varepsilon_1}, n_{\varepsilon_2}\}$ and $m \geq m_{\varepsilon} = \max\{m_{\varepsilon_1}, m_{\varepsilon_2}\}$, the performance guarantee is
\begin{align*}
(3 / 2 + \varepsilon_1)(\opt + |\overline{X}|) &< (3 / 2 + \varepsilon_1)(\opt + \varepsilon_2 \opt)\\
&\leq (3 / 2 + \varepsilon_1)(1 + \varepsilon_2) \opt\\
&= (3 / 2 + \varepsilon_1 + (3 / 2) \varepsilon_2 + \varepsilon_1 \varepsilon_2) \opt\\
&\leq (3 / 2 + \varepsilon) \opt
\end{align*}
\qed
\end{proof}

\section{Open questions}
\label{sec:open}

In this paper we only considered the unweighted case of \name{STAR}. If the input graph has weights, the problem obviously remains hard, in terms of finding both exact and approximate solutions. Unfortunately, for that case, the approximation strategy we proposed in Theorem \ref{th:star_approximation_algorithm} is no longer useful, because if the vertex cover is agnostic of the weights, then the constructed cycle may be forced to use heavy edges, and therefore the output can be made arbitrarily longer than an optimal solution. Is it possible to adapt the algorithm for the weighted case, or to devise a different constant-factor approximation algorithm?

On a separate note, we showed that there cannot be a PTAS for \name{STAR} unless $\class{P} = \class{NP}$. However, this doesn't rule out the possibility of a PTAS for the grid case, for which the best we have achieved is a $(3/2 + \varepsilon)$-approximation algorithm that only works for a proper subset of instances. Since the grid case is of practical interest, it would be worthwhile to investigate this possibility.

Finally, the problem may be extended in natural ways, like using multiple trucks  or considering the time it takes the driver to carry newspapers to the households.

\section*{Acknowledgements}

Thanks to Mart\'in Farach-Colton for useful discussions and suggestions about the presentation.

\bibliography{paper}

\end{document}